\documentclass[11pt]{article}
\usepackage[margin=1in]{geometry}
\usepackage{amsmath,amssymb,amsthm,mathtools}
\usepackage{booktabs}
\usepackage[colorlinks=true,linkcolor=blue,citecolor=blue]{hyperref}

\newtheorem{theorem}{Theorem}[section]
\newtheorem{proposition}[theorem]{Proposition}
\newtheorem{lemma}[theorem]{Lemma}
\newtheorem{corollary}[theorem]{Corollary}
\theoremstyle{definition}
\newtheorem{definition}[theorem]{Definition}
\newtheorem{example}[theorem]{Example}
\newtheorem{remark}[theorem]{Remark}

\newcommand{\R}{\mathbb{R}}
\newcommand{\C}{\mathbb{C}}
\newcommand{\N}{[n]}
\newcommand{\spn}{\mathrm{span}}
\newcommand{\Inv}{\mathrm{Inv}}

\title{What Multichoice Values Cannot See:\\ The Information Content of Anonymous Values\\ for Games with Graded Participation}
\author{Matthew Fried\thanks{Touro University.}}

\begin{document}
\maketitle

\begin{abstract}
In a cooperative game with graded participation, each of $n$ players acts at one of $m$ ordered levels; multichoice games, voting with abstention, and graded feature attribution all take this form. For this setting a substantial zoo of competing values, power indices, and importance measures has been proposed. We determine exactly what the entire family of linear, player-symmetric solutions, the class containing the standard proposals, can and cannot see, for all $m$ at once. The mechanism is short: the functional computing
any one player's payoff under an anonymous value is invariant under permutations of
the other $n-1$ players, and by the branching rule such invariants exist only in the
Specht constituents $(n)$ and $(n-1,1)$ of $(\C^m)^{\otimes n}$. Consequences: the
joint information of all anonymous values is a component of dimension polynomial in
$n$ against an $m^n$-dimensional game space, with the classical binary theory as the
$m=2$ shadow. \emph{Three players can hide}: for $m\ge3$ the blind space is nonzero already at $n=3$, and the first invisible constituent, arising at the minimal player count, is not a correlation but a chirality, realized by two distinct monotone abstention-voting rules on three voters that every anonymous power index scores identically. For every $n\ge4$ the blind space is spanned explicitly by $\pm1$ games supported on four profiles each. Order-$d$ interaction probes see exactly the partitions with at most $d$ cells outside the first row, with full recovery only at $d=n-\lceil n/m\rceil$. Audit evasion gets easier than in the binary theory: a coalition of $c$ players evades every order-$d$ audit iff $c-\lceil c/m\rceil\ge d+1$, so with three or more levels, trios can restructure invisibly to every value-based payment scheme. Algorithmically, the visible part of a game is a polynomial-size sketch, estimable from polynomially many queries, while exact computation of any full-support anonymous value requires querying all $m^n-1$ nonzero profiles. All dimension and rank claims are verified computationally.
\end{abstract}

\section{Introduction}

\subsection{The question}

The theory of values for games with graded participation is in an expansionary phase.
In a \emph{multichoice game}, each of $n$ players participates at one of $m$ ordered levels (out, part-time, full-time; off, attenuated, on), and a characteristic function records the worth of every configuration. Extensions of the Shapley value to
this setting have multiplied \cite{HR,vdN,KSZ,CS00,PZ05,GL07,Lowing} to the point that recent work
is devoted to comparing them and seeking consensus among them \cite{LT25}; the same
solutions are now deployed as feature- and parameter-importance measures in machine learning, where Shapley-style attribution is already the default instrument
\cite{LL17} and a component is not merely present or absent \cite{FLSN,GLR}; and a
parallel family of power indices serves voting with abstention and multiple levels of
approval \cite{FM,FZ,Freixas05,Tchantcho}. Each instrument answers the question
``how should credit be assigned?'' differently, and the literature compares them
axiom by axiom.

This paper asks a question that is prior to the axiomatics: whatever the assignment, what can the instruments \emph{see}? Which two games are distinguished by \emph{no} anonymous value, index, or
importance measure whatsoever, neither the published ones, nor future ones, nor all of them pooled together? Call the span of everything the family can measure the \emph{visible space}, and its orthogonal complement the \emph{blind space}: two games differing by a blind game receive identical treatment from every anonymous instrument at once. In computational terms, a game is an exponentially large function on $[m]^n$; an anonymous measurement may single out one coordinate and its symbol, but sees the remaining coordinates only through their histogram; and the question is the joint information content of all such linear measurements.

For binary participation ($m=2$) the answer is classical. Kleinberg and Weiss
decomposed binary game space under the symmetric group \cite{KW85a,KW85b,KW86};
Hern\'andez-Lamoneda, Ju\'arez and S\'anchez-S\'anchez showed that every symmetric
linear value factors through the two constituents isomorphic to summands of $\R^n$
\cite{HLJS}; and Amer, Derks and Gim\'enez constructed the common kernel of the
semivalue family, proving it nonzero exactly from four players and spanning it by
explicit $\{-1,0,1\}$-valued games \cite{ADG} (see \cite{CO} for a survey). For $m\ge3$ the closest neighbor is the representation-theoretic study by Saito and Kusunoki \cite{SK22} of \emph{bicooperative} games \cite{Bilbao00}, a ternary class. In the spirit of \cite{HLJS}, they decompose that game space slice by slice, parametrize all linear symmetric solutions, and prove that the unanalyzed remainder of their decomposition is annihilated by every such solution. For $m=3$, order one, and payoffs in $\R^n$, the joint-information question is thereby answered implicitly, and even the nonvanishing of the blind space at three players can be extracted from their theorem by a dimension count, although it is not remarked there. What their analysis deliberately sets aside (their Remark~3.2 declares the remainder to play no role) is the blind space itself: its content, its size, its witnesses, its behavior across $m$, and its response to stronger probes. Those questions are this paper's subject, for every $m$ at once.

We compute the answer for every $m$ simultaneously, and it is not a routine
transcription of the binary case. The threshold for invisibility drops from four players to three; the first invisible constituent changes character entirely, from a correlation pattern to a chirality; the ladder of interaction probes needed for full recovery grows taller; and audit evasion gets strictly easier. Table~\ref{tab:compare} is the summary; provenance for each entry is stated precisely in \S\ref{subsec:technique}.

\paragraph{Significance.}
Three points position the paper. First, the main theorem is an impossibility
result for a class of algorithms rather than a comparison of solutions within
it. Anonymity is a structural constraint on estimators, and
Theorem~\ref{thm:main} computes the exact information ceiling that this
constraint imposes, for every level count at once; the bound therefore applies
to every instrument in the cited literatures and to every instrument not yet
invented. Second, the ceiling has a query-complexity reading, made exact in Section~\ref{sec:algo}. The visible space is a linear sketch of dimension $O(n^m)$ of an $m^n$-dimensional object, of minimal rank among sketches determining every anonymous value; everything an anonymous auditor can see is estimable from polynomially many queries, while exact computation of any full-support anonymous value requires querying all $m^n-1$ nonzero profiles. The blind space supplies the matching oracle separation: a single profile query distinguishes games that all anonymous instruments jointly cannot, at any sample size. Third, the proof explains constants instead of merely computing them. The binary analysis yields three benchmark thresholds: the classical four-player invisibility threshold \cite{ADG} and, from the higher-order binary study \cite{Fried26}, the $\lfloor n/2\rfloor$ recovery cap and the $c\ge2d+2$ evasion bound. A single
application of the branching rule shows that all three are two-row facts, that
the number of rows available equals the number of participation levels, and
hence exactly how each constant moves when $m$ grows. The consequences for the
motivating literatures follow at once: graded feature attribution has a
strictly larger blind spot than binary attribution at the same probe order, and
three colluding agents can restructure invisibly to every value-based payment scheme as soon as participation has three levels.

\begin{table}[h]
\centering
\begin{tabular}{@{}lll@{}}
\toprule
 & binary ($m=2$) & graded ($m\ge3$)\\
\midrule
invisibility threshold & $n\ge4$ \cite{ADG} & $n\ge3$ (Thm~\ref{thm:three})\\
first invisible constituent & correlation & chirality (Ex.~\ref{ex:chirality})\\
spanning the blind space & $\pm1$ on four coalitions \cite{ADG} & $\pm1$ on four profiles (Thm~\ref{thm:span})\\
recovery cap, order-$d$ probes & $\lfloor n/2\rfloor$ \cite{Fried26} & $n-\lceil n/m\rceil$ (Thm~\ref{thm:ladder})\\
invisible evasion vs.\ order-$d$ audit & $c\ge2d+2$ \cite{Fried26} & $c-\lceil c/m\rceil\ge d+1$ (Cor~\ref{cor:manip})\\
visible dimension & $n^2-n+1$ & degree-$m$ polynomial in $n$ (Thm~\ref{thm:dims})\\
\bottomrule
\end{tabular}
\caption{The binary theory versus the graded theory, at a glance. Nonvanishing of the $m=3$ blind space at $n=3$ (first row) is extractable by dimension count from \cite{SK22}, as discussed in the introduction; the identification of its content, and the remaining graded entries, are new here.}
\label{tab:compare}
\end{table}

\subsection{Results}

Throughout, an \emph{anonymous value} is any linear solution concept that treats
players symmetrically: payoffs may depend on a player's own level and on the census of
everyone else's levels, but never on anyone's name (Definition~\ref{def:anonymous}).
Lemma~\ref{lem:anonymous} verifies that this class contains the standard multichoice
values, indices, and importance measures, including those that assign payoffs to player--level pairs, while Section~\ref{sec:setup} delimits what falls outside (player-weighted, heterogeneous-level, and nonlinear concepts).

\begin{enumerate}
\item \textbf{Characterization} (Theorem~\ref{thm:main}, \S\ref{sec:main}). The
visible space is exactly the $(n)\oplus(n{-}1,1)$ isotypic component of
$(\C^m)^{\otimes n}$ under the symmetric group permuting players. One branching-rule argument proves this for every $m$ at once; the classical binary theory drops out as the case $m=2$, and the ternary order-one statement implicit in \cite{SK22} as the case $m=3$.
\item \textbf{Dimensions} (Theorem~\ref{thm:dims}, \S\ref{sec:dims}). The visible
dimension is $\binom{n+m-1}{m-1}+(n-1)^2\binom{n+m-2}{n}$, a polynomial in $n$ of degree $m$, against a game space of dimension $m^n$. At fixed $n\ge3$ the blind \emph{fraction} tends to $1-\bigl(1+(n-1)^2\bigr)/n!$ as $m$ grows: richer participation reduces, not raises, the completeness of valuation.
\item \textbf{Three can hide, and what hides} (Theorems~\ref{thm:three}
and~\ref{thm:span}, \S\ref{sec:three}). The binary blind space is zero until four
players; for $m\ge3$ it is nonzero already at three, spanned by determinant games, which are chiralities (Example~\ref{ex:chirality}). And for every $n\ge4$ the entire blind space
is spanned by \emph{double-swap games}: $\pm1$ games supported on four profiles
related by exchanging two players' levels in one place and two other players' levels
in another. At $m=2$ the double-swap games span the classical semivalue kernel,
yielding a family of the same four-coalition kind as \cite{ADG}.
\item \textbf{The ladder} (Theorem~\ref{thm:ladder}, \S\ref{sec:ladder}). Order-$d$
interaction probes see exactly the Specht constituents with at most $d$ cells outside
the first row; full recovery arrives only at $d=n-\lceil n/m\rceil$, never earlier than the binary $\lfloor n/2\rfloor$ and strictly later for every $n\ge3$ except $(n,m)=(4,3)$, where the caps coincide.
\item \textbf{Audit evasion} (Corollary~\ref{cor:manip}, \S\ref{sec:ladder}). A coalition of $c$ players can restructure how value depends on its own levels, invisibly to every order-$\le d$ payment scheme, iff $c-\lceil c/m\rceil\ge d+1$. With three or more levels, trios evade every value-based scheme.
\end{enumerate}

Section~\ref{sec:applications} converts these into statements for the three motivating
literatures, including two distinct monotone $(3,2)$ voting rules on three voters that every anonymous index scores identically (Proposition~\ref{prop:voting}), and
Section~\ref{sec:algo} gives the query-complexity reading: the visible part of a game is polynomial-size, cheaply estimable, and of minimal sketch rank, while the game itself is not.

\subsection{An invisible chirality, concretely}

The following example contains, in miniature, most of what is new in this paper.

\begin{example}[three players hide a chirality]\label{ex:chirality}
Three consultants $A,B,C$ each work at level $0$ (out), $1$ (part-time), or $2$
(full-time). Writing a profile as $(x_A,x_B,x_C)$, define a game $\varepsilon$ on the
$27$ configurations: it is $0$ except
on the six profiles where all three levels differ, where
\[
\varepsilon(0,1,2)=\varepsilon(1,2,0)=\varepsilon(2,0,1)=+1,\qquad
\varepsilon(0,2,1)=\varepsilon(2,1,0)=\varepsilon(1,0,2)=-1.
\]
The game rewards one \emph{orientation} of the division of labor, the rotation in which $A$ hands off to $B$ hands off to $C$, and penalizes its mirror image. Now
aggregate the way any anonymous instrument must: fix one player's level and the census
of the others. Every such statistic pairs each profile with its swapped partner and
cancels; for instance, ``$A$ out, the others at part- and full-time'' reads
\[
\varepsilon(0,1,2)+\varepsilon(0,2,1)=1-1=0,
\]
and the same happens for every player and every census. So every anonymous value,
power index, or importance measure, every one in the literature and every one yet to be proposed, treats $\varepsilon$ exactly as it treats the zero game, even
though the two orientations it separates can be economically opposite. With binary
participation nothing of this kind exists below four players
(Theorem~\ref{thm:three}); and up to scale, this six-profile game is
\emph{everything} three players with three levels can hide.
\end{example}

\subsection{Technique, related work, and what is not claimed}\label{subsec:technique}

The entire mechanism is one observation plus one classical fact. Anonymity means the
functional computing player $i$'s payoff is invariant under the symmetric group on the
other $n-1$ players; the branching rule says which irreducible constituents of game
space admit such invariants: only $(n)$ and $(n-1,1)$. Everything else in the paper (dimension counts, thresholds, spanning families, the ladder, the evasion bounds) is the bookkeeping this one observation makes possible. We regard the brevity as the point: one mechanism accounts for every constant in Table~\ref{tab:compare}, and it bounds instruments that have not been invented yet.

To be precise about foundation versus contribution. \emph{Known}: the binary theory recalled in Section~\ref{sec:binary-recap} is due to \cite{KW85a,KW85b,KW86,HLJS,ADG}; for $m=3$ with payoffs in $\R^n$, the order-one characterization is implicit in \cite{SK22} as described above, obtained there by codomain-side Schur arguments in the style of \cite{HLJS}; and the binary order-$d$ cap $\lfloor n/2\rfloor$ and evasion threshold $c\ge2d+2$ come from the binary study \cite{Fried26}, of which the present paper is independent. \emph{New here}: the uniform characterization for every $m$ through domain-side invariants (the readings are $S_{n-1}$-invariant vectors, and the branching rule does the rest), which is what extends to order-$d$ probes where codomain arguments do not; the identification of the blind content, beginning with the chirality $\Lambda^3\C^m$ and the count $\binom m3$; the dimension formulas; the double-swap spanning theorem; the full ladder and its cap; the coalition-local evasion threshold; the monotone $(3,2)$ collision; and coverage of solutions valued in player--level pairs, outside the $\R^n$ codomain of \cite{SK22}. All ranks and dimensions are verified by explicit computation (Appendix~\ref{app:verification}).

\section{The binary foundations}\label{sec:binary-recap}

\emph{Everything in this section is known. We restate the binary theory in the
notation the rest of the paper will reuse, following \cite{KW85a,ADG,HLJS,CO}; a
reader who knows these results can move directly to Section~\ref{sec:setup}.}

Binary game space
$\{v:2^{\N}\to\R,\,v(\varnothing)=0\}$ decomposes by cardinality slices, each slice
$\mathcal G_k\cong M^{(n-k,k)}$ splitting into Specht modules
$S^{(n)}\oplus S^{(n-1,1)}\oplus\cdots\oplus S^{(n-\min(k,n-k),\min(k,n-k))}$. By
Schur's lemma, every symmetric linear value, semivalues \cite{DNW} included, factors through the constituents isomorphic to summands of $\R^n$, i.e.\ $(n)$ and $(n-1,1)$;
the common kernel of the semivalue family is therefore the sum of all constituents
$S^{(n-j,j)}$ with $j\ge2$, is nonzero precisely from $n=4$, and admits an explicit
spanning family of $\{-1,0,1\}$-valued games \cite{ADG}. The visible dimension is
$n+(n-1)^2=n^2-n+1$.

\section{Multichoice games and anonymous values}\label{sec:setup}

\emph{Before proving anything we must fix what ``every anonymous value'' means, and
the definition has to carry real weight: the main theorem bounds a whole class, so the
class had better contain the values people actually use. This section sets up the
model, introduces the reading $F_{i,\ell,\tau}$, the census statistic from which every anonymous payoff is built, and proves a coverage lemma pinning down exactly
which values the bound applies to.}

\begin{definition}
Fix $n\ge2$ players and $m\ge2$ \emph{participation levels} $\{0,1,\dots,m-1\}$, level
$0$ meaning non-participation. A \emph{multichoice game} is a function
$v:\{0,\dots,m-1\}^n\to\R$ with $v(\mathbf 0)=0$. Game space $G_{n,m}$ has dimension
$m^n-1$; we work in the full $m^n$-dimensional space $\C[\{0,\dots,m-1\}^n]\cong
(\C^m)^{\otimes n}$ and account for the anchor at the end; unless stated otherwise,
dimension statements refer to the full space.\footnote{We work over $\C$
for convenience only; every module in this paper is defined over $\mathbb{Q}$ and nothing
depends on the field.} The symmetric group $S_n$ acts by permuting players (tensor slots); levels are \emph{not} permuted, so they remain distinguished labels. The order on levels plays no role in the representation theory, which therefore applies to any common finite action alphabet; the order enters only through the monotonicity and voting applications.
\end{definition}

\begin{definition}\label{def:anonymous}
For a player $i$, an \emph{aggregated reading} is determined by a pair $(\ell,\tau)$
with $\ell$ a level and $\tau$ a \emph{type} (a multiset recording how many of the
other $n-1$ players sit at each level):
\[
F_{i,\ell,\tau}(v)\;=\;\sum_{\substack{x:\ x_i=\ell,\\ \mathrm{type}(x_{-i})=\tau}} v(x).
\]
Equivalently, writing $r_{i,\ell,\tau}=\sum_{x:\ x_i=\ell,\ \mathrm{type}(x_{-i})=\tau}e_x$ for the orbit-sum vectors and $\langle\cdot,\cdot\rangle$ for the inner product making the profile basis orthonormal, $F_{i,\ell,\tau}(v)=\langle r_{i,\ell,\tau},v\rangle$. An \emph{anonymous value} is any value $\psi$ such that each $\psi_i$ is a linear combination of the readings $\{F_{i,\ell,\tau}\}_{\ell,\tau}$ with coefficients independent of $i$ (equivalently, by Lemma~\ref{lem:anonymous}: $\psi$ is linear and equivariant under relabeling players). The \emph{joint information} of the class is the visible space $\mathcal V_{n,m}=\spn\{r_{i,\ell,\tau}:\ i,\ell,\tau\}\subseteq(\C^m)^{\otimes n}$, and the \emph{blind space} is its orthogonal complement $\mathcal N_{n,m}=\mathcal V_{n,m}^{\perp}$: two games receive identical treatment from every anonymous value iff they differ by an element of $\mathcal N_{n,m}$.
\end{definition}

An anonymous value may thus pay player $i$ for $i$'s own level and for the census of
everyone else, meaning how many others sit at each level, but never for \emph{which} others. Concretely, with $n=4$ and $m=3$, one reading of player $1$ sums $v$ over the three profiles in which player $1$ is at level $2$ while, among the other three players, one sits at level $0$ and two sit at level $1$. The following lemma is the linchpin of every coverage claim in the paper: the
anonymity class is exactly the equivariant linear values, including those that assign
payoffs to player--level pairs, which is how most multichoice values are formulated.

\begin{lemma}[coverage]\label{lem:anonymous}
A linear value $\psi:(\C^m)^{\otimes n}\to\C^n$ is equivariant under relabeling
players iff each $\psi_i$ is a linear combination of the readings
$\{F_{i,\ell,\tau}\}_{\ell,\tau}$ with coefficients independent of $i$. The same holds
componentwise for values $\psi:(\C^m)^{\otimes n}\to\C^{n\times(m-1)}$ assigning
payoffs to player--level pairs, as in \cite{vdN,HR,Lowing}.
\end{lemma}

\begin{proof}
If $\sigma\in S_n$ fixes $i$, equivariance gives $\psi_i(\sigma v)=\psi_i(v)$: the
functional $\psi_i$ is invariant under the subgroup $S_{n-1}^{(i)}$ permuting the
other players, whose orbits on level profiles are exactly the fibers of
$x\mapsto(x_i,\mathrm{type}(x_{-i}))$; hence $\psi_i$ is a linear combination of the
orbit-indicator functionals, which are the readings $F_{i,\ell,\tau}$. Equivariance
under general $\sigma$ forces the coefficients to be independent of $i$. The converse
is a direct check. For player--level pairs, apply the argument to each functional
$\psi_{i,j}$ separately: equivariance acts on the player index only, and each
$\psi_{i,j}$ is again $S_{n-1}^{(i)}$-invariant.
\end{proof}

The class is broad but has boundaries, and we state both. \emph{Inside}: any value
computed linearly from a player's levels and the census of the others. This covers the values of \cite{vdN,KSZ,CS00,PZ05,Lowing}, the symmetric members of the Sharing-value families
\cite{LY25}, the abstention and multichoice power indices \cite{FM,FZ,Freixas05,Tchantcho,Freixas20}, the bisemivalues on bicooperative games \cite{DGP20} (and, insofar as they are linear and player-symmetric, newly proposed bicooperative solutions such as the 2-Shapley value \cite{BGGJ26}), the importance indices of \cite{GLR} (see \cite{Grabisch16} for the $k$-ary
framework), the signed interaction index of \cite{RGL20} (an order-$d$ probe; see
\S\ref{sec:ladder}), and the feature-importance solutions of \cite{FLSN}.
Theorem~\ref{thm:main} upper-bounds all of them simultaneously, present and future.
\emph{Outside}: Hsiao--Raghavan values \cite{HR} with \emph{player-dependent} weight systems, which exit the symmetric framework exactly as weighted Shapley values do in the binary world (player-independent action weights, by contrast, stay inside: the coefficients in Definition~\ref{def:anonymous} may depend on levels arbitrarily); frameworks in which players have heterogeneous level counts
$m_i$, unless levels are equalized; and nonlinear or set-valued concepts, such as the absolute interaction index of \cite{RGL20} and the core, to which linear-information bounds do not directly apply.

\section{The characterization}\label{sec:main}

\emph{We now prove the paper's central result. One theorem identifies the joint
information of every anonymous value at every level count, its proof is a single
application of the branching rule, and the binary theory of \cite{KW85a,ADG,HLJS} is
recovered as the case $m=2$.}

We use standard symmetric-group representation theory \cite{Sagan}: irreducible
$S_n$-modules $S^\lambda$ are indexed by partitions $\lambda\vdash n$;
$(\C^m)^{\otimes n}$ contains $S^\lambda$ iff $\lambda$ has at most $m$ rows, with
multiplicity $\dim_{GL_m}\lambda$ (Schur--Weyl duality \cite{FH}); and the branching rule states that $S^\lambda$ restricted to $S_{n-1}$ is the sum of $S^\mu$ over all $\mu$ obtained by removing one corner cell of $\lambda$.

The entire paper rests on the following lemma, stated once for all probe orders; the main theorem is its order-one case, and Section~\ref{sec:ladder} will use it at full strength.

\begin{lemma}[symmetry--visibility lemma]\label{lem:sv}
Fix $0\le d\le n$ and, for a $d$-element set $D\subseteq\N$, let $S_{n-d}^{(D)}$ be the subgroup of $S_n$ fixing $D$ pointwise. The $S_n$-module generated by the $S_{n-d}^{(D)}$-invariant vectors of $(\C^m)^{\otimes n}$ (the same module for every choice of $D$, hence equal to the span of the invariants over all $D$) is the isotypic component of the partitions $\lambda$ with at most $d$ cells outside the first row (and at most $m$ rows).
\end{lemma}

\begin{proof}
Decompose $(\C^m)^{\otimes n}\cong\bigoplus_\lambda S^\lambda\otimes M_\lambda$ with multiplicity spaces $M_\lambda$ of dimension $\dim_{GL_m}\lambda$. The $S_{n-d}^{(D)}$-invariant vectors form $\bigoplus_\lambda \Inv_{S_{n-d}}(S^\lambda)\otimes M_\lambda$, and iterating the branching rule $d$ times, $\Inv_{S_{n-d}}(S^\lambda)\neq0$ iff $(n-d)$ is reachable from $\lambda$ by removing $d$ corner cells, iff $\lambda$ has at most $d$ cells outside its first row. This proves containment. For the converse, multiplicities matter: an $S_n$-submodule of an isotypic component $S^\lambda\otimes M_\lambda$ need not be the whole component, since it can be $S^\lambda\otimes M'$ for a proper subspace $M'\subsetneq M_\lambda$, so it does not suffice to ``meet'' each constituent. But the invariant space itself already contains $u\otimes M_\lambda$ for any nonzero $u\in\Inv_{S_{n-d}}(S^\lambda)$: the multiplicity space is exhausted before generation begins. The $S_n$-module generated by $u\otimes M_\lambda$ is $\langle S_n\,u\rangle\otimes M_\lambda=S^\lambda\otimes M_\lambda$ by irreducibility of $S^\lambda$; and both containment and generation are independent of the choice of $D$, since the subgroups $S_{n-d}^{(D)}$ are conjugate in $S_n$.
\end{proof}

\begin{theorem}[characterization]\label{thm:main}
For every $n\ge2$ and $m\ge2$,
\[
\mathcal V_{n,m}\;=\;\text{the }\bigl(S^{(n)}\oplus S^{(n-1,1)}\bigr)\text{-isotypic
component of }(\C^m)^{\otimes n}.
\]
\end{theorem}

\noindent\emph{Proof idea. The readings of a single player are the orbit sums spanning the full space of invariants under permutations of the other players; the lemma converts that invariance into the two-row characterization.}

\begin{proof}
Fix $i$. The orbits of $S_{n-1}^{(i)}$ on level profiles are precisely the fibers of $x\mapsto(x_i,\mathrm{type}(x_{-i}))$, so the reading vectors $r_{i,\ell,\tau}$ are the orbit-sum basis of the \emph{full} space of $S_{n-1}^{(i)}$-invariant vectors. Since $\mathcal V_{n,m}$ is an $S_n$-submodule (permuting players permutes readings) and every reading is an $S_n$-translate of a reading of player $i$, it is exactly the $S_n$-module generated by these invariants, and Lemma~\ref{lem:sv} with $d=1$ identifies it: the partitions with at most one cell outside the first row are $(n)$ and $(n-1,1)$. (Concretely, for $\lambda=(n-1,1)$ realized as the sum-zero subspace of $\C^n$, the invariant line attached to player $i$ is spanned by $e_i-\tfrac1n\sum_j e_j$, and these vectors span the module as $i$ varies.)
\end{proof}

Pooled, the anonymous instruments see the trivial-plus-standard part of game space and nothing else.

\begin{remark}[the binary case, and where it sits]
At $m=2$, Theorem~\ref{thm:main} specializes to the known statements of
\cite{KW85a,ADG,HLJS}: symmetric linear values see the $(n)\oplus(n-1,1)$ part of
binary game space, and the common kernel is everything else. The content of the theorem is that this is not a fact about subsets: it is a fact about $S_{n-1}$-invariance, valid verbatim for every level structure, with the branching rule as the entire mechanism. The binary theory is the $m=2$ shadow; the ternary order-one statement, with payoffs in $\R^n$, is similarly implicit in the slice-by-slice analysis of \cite{SK22}, and the lemma replaces those slice constructions with one branching computation valid for every $m$ and every probe order.
\end{remark}

\begin{remark}[not a low-degree statement]
Readers accustomed to Fourier or Efron--Stein decompositions of functions on product
spaces should note that visibility is governed by \emph{symmetry type}, not degree.
The visible space contains fully symmetric functions of every degree up to $n$;
conversely, the chirality of Example~\ref{ex:chirality}, blind in its entirety, has
vanishing degree-one part but nonzero Efron--Stein components already at degree two.
The two decompositions are transverse.
\end{remark}

\section{Dimensions}\label{sec:dims}

\emph{With the visible space identified, the next question is its size. This section
computes the dimension exactly and draws the comparison that drives the rest of the
paper: the visible space grows polynomially in $n$, the game space exponentially, and
raising $m$ widens the gap.}

\begin{theorem}[dimensions]\label{thm:dims}
$\dim_{GL_m}(n)=\binom{n+m-1}{m-1}$ and $\dim_{GL_m}(n-1,1)=(n-1)\binom{n+m-2}{n}$.
Hence
\[
\dim\mathcal V_{n,m}\;=\;\binom{n+m-1}{m-1}\;+\;(n-1)^2\binom{n+m-2}{n},
\]
polynomial in $n$ of degree $m$, and the blind space
$\mathcal N_{n,m}$ has dimension $m^n-\dim\mathcal V_{n,m}$.
(The anchor $v(\mathbf 0)=0$ removes one \emph{visible} dimension, because the all-inactive profile is $S_n$-invariant and hence visible, and leaves the blind dimension unchanged.) For $m=3$ the blind dimensions
at $n=2,\dots,8$ are $0,\,1,\,21,\,126,\,526,\,1863,\,6075$; for $m=2$ the formula
reproduces the classical $2^n-n^2+n-2$.
\end{theorem}

\begin{proof}
$\dim_{GL_m}(n)$ counts weakly increasing words (symmetric power). For $(n-1,1)$ the
hook content formula \cite{Stanley} gives
$\prod_{c}(m+\mathrm{content}(c))/\mathrm{hook}(c)
=(m-1)\,m(m+1)\cdots(m+n-2)\big/\bigl(n\cdot(n-2)!\bigr)$, which simplifies to
$(n-1)\binom{n+m-2}{n}$ (at $m=3$: $n^2-1$; at $m=2$: $n-1$). The multiplicity of
$S^{(n-1,1)}$ contributes a factor $\dim S^{(n-1,1)}=n-1$. The numerical sequences are
verified in Appendix~\ref{app:verification}, including by explicit measurement-matrix
ranks ($26/27$ at $n=3$, $60/81$ at $n=4$ for $m=3$).
\end{proof}

\begin{remark}[expressiveness backfires]
Two comparisons quantify the gap. At fixed $m$, the visible dimension grows polynomially in $n$ while the game space grows exponentially. At fixed $n\ge3$, both are polynomials of degree $n$ in $m$, with leading coefficients $\bigl(1+(n-1)^2\bigr)/n!$ and $1$ respectively, so the visible \emph{fraction} converges to $\bigl(1+(n-1)^2\bigr)/n!$ as $m\to\infty$: at $n=8$, below $0.13\%$. (At $n=2$ the blind space is zero for every $m$.) Upgrading a data market from include/exclude to quality tiers makes valuations less complete unless the probes upgrade too (Section~\ref{sec:applications}).
\end{remark}

\section{What hides: three players, and an explicit spanning family}\label{sec:three}

\emph{The best-known constant of the binary theory is its four-player threshold: with
three players, nothing hides \cite{ADG}. This section shows the threshold is an
artifact of binarity, and then answers the harder question, not how large the blind space is but what is in it, with an explicit family of four-profile $\pm1$ games that spans it for every $n\ge4$, the determinant games covering $n=3$.}

\begin{theorem}[three can hide]\label{thm:three}
For $m\ge3$, $\dim\mathcal N_{3,m}=\binom m3>0$: the blind space is nonzero already at
three players, and it is spanned by \emph{determinant games}: for each triple of
levels $a<b<c$, the game
\[
\varepsilon_{abc}(\ell_1,\ell_2,\ell_3)=
\begin{cases}
\operatorname{sgn}(\sigma) & \text{if }(\ell_1,\ell_2,\ell_3)=\sigma(a,b,c)
\text{ for a permutation }\sigma,\\
0 & \text{otherwise.}
\end{cases}
\]
For $m=2$ the blind space at $n=3$ is zero (the classical four-player threshold
\cite{ADG}).
\end{theorem}

\begin{proof}
At $n=3$ the partitions outside $\{(3),(2,1)\}$ with at most $m$ rows are exactly
$(1,1,1)$, present iff $m\ge3$, contributing $S^{(1,1,1)}\otimes\Lambda^3\C^m$ of
dimension $1\cdot\binom m3$. The sign-isotypic vectors of $(\C^m)^{\otimes3}$ are
spanned by the antisymmetrizations of $e_a\otimes e_b\otimes e_c$ over level-triples,
which are precisely the games $\varepsilon_{abc}$. Verified: every aggregated reading
annihilates $\varepsilon_{012}$ exactly (Appendix~\ref{app:verification}).
\end{proof}

\begin{remark}[correlation versus chirality]
Example~\ref{ex:chirality} displays $\varepsilon_{012}$ in full. The contrast with the
binary case is worth naming: the minimal binary camouflage is a \emph{correlation}, a four-cycle of pairwise synergies that needs four players to balance \cite{ADG}, while the minimal graded camouflage is a \emph{chirality}, balanced by three players
because antisymmetry, not pairing, does the canceling. With binary participation,
payment schemes cannot see who pairs with whom; with graded participation they also
cannot see which way the division of labor turns.
\end{remark}

What hides at general $n$? The next theorem answers with an object as concrete as the
chirality. For a profile $w$ and a pair of players $\{a,b\}$, write $(ab)w$ for the
profile with the levels of $a$ and $b$ exchanged, and $\delta_x$ for the game worth
$1$ at profile $x$ and $0$ elsewhere.

\begin{theorem}[spanning family]\label{thm:span}
For every $n\ge4$ and $m\ge2$, the blind space $\mathcal N_{n,m}$ is spanned by the
\emph{double-swap games}
\[
D_{w;ab,cd}\;=\;\delta_w-\delta_{(ab)w}-\delta_{(cd)w}+\delta_{(ab)(cd)w},
\]
where $\{a,b\}$ and $\{c,d\}$ are disjoint pairs of players and $w$ is any profile
with $w_a\neq w_b$ and $w_c\neq w_d$. Each $D_{w;ab,cd}$ takes values $\pm1$ on four
profiles and $0$ elsewhere. For $n=3$ the blind space is spanned by the determinant
games of Theorem~\ref{thm:three}. At $m=2$ the double-swap games span the classical
common kernel of the semivalues, giving a four-coalition $\{\pm1\}$ family of the same
kind as \cite{ADG}.
\end{theorem}

\noindent\emph{Proof idea. A double-swap game is a column antisymmetrization for the
shape $(n-2,2)$, and $(n-2,2)$ dominates every blind shape; Kostka positivity turns
dominance into spanning.}

\begin{proof}
Let $T$ be a filling of the diagram of $\lambda=(n-2,2)$ whose two columns of height
$2$ contain $\{a,b\}$ and $\{c,d\}$, and let
$\kappa_T=\sum_{\sigma\in C_T}\mathrm{sgn}(\sigma)\,\sigma$ be its column
antisymmetrizer, where the column group is the Young subgroup
$C_T=S_{\{a,b\}}\times S_{\{c,d\}}\cong S_{\lambda'}$ for
$\lambda'=(2,2,1,\dots,1)$. Then $\kappa_T\,\delta_w=D_{w;ab,cd}$, and
$\kappa_T\,\delta_w=0$ when $w_a=w_b$ or $w_c=w_d$; so as $w$ varies, the double-swap
games for the fixed pair of pairs span exactly the image
$\kappa_T\bigl((\C^m)^{\otimes n}\bigr)$.

Since $\kappa_T$ is $|C_T|$ times the $C_T$-equivariant projection onto the
$\mathrm{sgn}_{C_T}$-isotypic subspace (a projection for the column-subgroup action, not an $S_n$-central one), its image meets each component
$S^\mu\otimes M_\mu$ in $W_\mu\otimes M_\mu$, where $W_\mu$ is the
$\mathrm{sgn}$-isotypic part of $\mathrm{Res}_{C_T}S^\mu$. Twisting by the sign
character,
\[
\mathrm{Hom}_{C_T}\!\bigl(\mathrm{sgn},\mathrm{Res}\,S^\mu\bigr)\;\cong\;
\mathrm{Hom}_{C_T}\!\bigl(\mathbf 1,\mathrm{Res}\,(S^\mu\otimes\mathrm{sgn})\bigr)
\;\cong\;
\mathrm{Hom}_{C_T}\!\bigl(\mathbf 1,\mathrm{Res}\,S^{\mu'}\bigr),
\]
which by Young's rule has dimension the Kostka number $K_{\mu'\lambda'}$, nonzero iff
$\mu'\trianglerighteq\lambda'$, iff $\mu\trianglelefteq(n-2,2)$ \cite{Sagan}.
Dominance $\mu\trianglelefteq(n-2,2)$ holds iff $\mu_1\le n-2$: necessity is the first
partial sum, and sufficiency is immediate since all later partial sums of $(n-2,2)$
equal $n$. So the image of $\kappa_T$ is supported on exactly the blind shapes
$\{\mu:\mu_1\le n-2\}$, with the full multiplicity space $M_\mu$ over each.

Thus every double-swap game is blind, and the fixed pair of pairs already exhausts
each multiplicity space: the image contains $u\otimes M_\mu$ with $u\neq0$ for every
blind $\mu$. The span over all choices of $\{a,b\},\{c,d\}$ is an $S_n$-submodule
(relabeling players permutes the family), and as in the proof of
Theorem~\ref{thm:main}, irreducibility fills each component, giving all of
$\mathcal N_{n,m}$.

For $n=3$ the shape $(n-2,2)$ does not exist; the unique blind shape is $(1,1,1)$,
whose column antisymmetrizer is the full alternating sum, producing the determinant
games. At $m=2$ a profile with $w_a\neq w_b$ and $w_c\neq w_d$ is a coalition
containing exactly one of $a,b$ and exactly one of $c,d$; the four support profiles
are the four equal-size coalitions obtained by exchanging memberships. The spanning
claim is verified by rank computation in eleven configurations
(Appendix~\ref{app:verification}).
\end{proof}

\begin{remark}[what invisibility looks like]
A double-swap game is the atom of invisible restructuring: take any configuration,
exchange the roles of two agents in one place and of two other agents in another, and
alternate the sign. Every anonymous instrument averages over exactly such exchanges,
so the alternation cancels. The theorem says these four-profile atoms not only hide but \emph{generate} everything that hides. The census sees how many sit at each level; it can never see who traded places with whom. At $m=2$ the same kernel now has three explicit spanning families: the shuffle games of \cite{ADG}, the dividend-coordinate four-cycle moves of \cite{Fried26}, and the coalition-coordinate double-swap games above.
\end{remark}

\section{The hierarchy and its cap}\label{sec:ladder}

\emph{An anonymous value is an order-$1$ probe: it examines one player at a time
against the census of the rest. This section asks what stronger probes buy. We climb
to order-$d$ readings, which examine $d$ players jointly, identify exactly what each
rung of the ladder sees, and then turn the ladder around: the same theorem that tells
auditors what order they need tells coalitions how large they must be to escape it. The proof of the evasion corollary is in Appendix~\ref{app:manip}.}

Order-$d$ readings distinguish a set $D$ of $d$ players jointly: for a level profile
$\boldsymbol\ell\in\{0,\dots,m-1\}^D$ and a type $\tau$ of the rest,
$F_{D,\boldsymbol\ell,\tau}(v)=\sum_{x:\ x_D=\boldsymbol\ell,\ \mathrm{type}(x_{-D})
=\tau}v(x)$. These are the natural multichoice analogues of the interaction functionals of
\cite{GrabischRoubens}; in the multichoice setting itself, the signed interaction
indices axiomatized by \cite{RGL20} are order-$d$ probes in exactly this sense, so
Theorem~\ref{thm:ladder} bounds them too.

\begin{theorem}[the ladder]\label{thm:ladder}
The joint span of all order-$\le d$ readings is the isotypic component of
$\{\lambda:\ |\lambda|-\lambda_1\le d,\ \ell(\lambda)\le m\}$. It equals the full
space iff $d\ge n-\lceil n/m\rceil$.
\end{theorem}

\begin{proof}
For fixed $D$ the order-$d$ readings are the orbit sums of the $S_{n-d}^{(D)}$-action on level profiles, hence span the full space of $S_{n-d}^{(D)}$-invariant vectors, and every reading of order $d'<d$ is a sum of order-$d$ readings, so the joint span of orders $\le d$ is the $S_n$-module generated by those invariants. Lemma~\ref{lem:sv} identifies it as the isotypic component of $\{\lambda:\ |\lambda|-\lambda_1\le d,\ \ell(\lambda)\le m\}$. The cap: $\max\{|\lambda|-\lambda_1:\ \ell(\lambda)\le m\} =n-\lceil n/m\rceil$, attained by the most balanced shape. At $m=2$ this is the binary $\lfloor n/2\rfloor$; at $m=3$, $n=5$ the rungs are $117\to228\to243$, verified by explicit ranks.
\end{proof}

For the corollary, write $\mathcal N^{(d)}_{n,m}$ for the \emph{order-$d$ blind
space}, the orthogonal complement of the span of all order-$\le d$ readings; thus
$\mathcal N^{(1)}_{n,m}=\mathcal N_{n,m}$.

\begin{corollary}[audit evasion]\label{cor:manip}
A coalition of $c$ players restructuring only the dependence of the game on its own
levels has, against order-$\le d$ payment schemes, an invisible-restructuring space
isomorphic to $\mathcal N^{(d)}_{c,m}$; in particular, invisible restructurings exist
iff $c-\lceil c/m\rceil\ge d+1$. For $m\ge3$, trios evade every order-$1$ (value-based) scheme; at $m=2$ the thresholds $c\ge2d+2$ of the binary analysis \cite{Fried26} are recovered. The minimal evading coalition is non-increasing in $m$: it drops from $2d+2$ at $m=2$ and saturates at $d+2$ once $m\ge d+2$.
\end{corollary}

\noindent\emph{Proof idea (full proof in Appendix~\ref{app:manip}). Blindness
transfers down: an $n$-player reading of a coalition-local perturbation factors, by a
completion count, through $c$-player readings of the same order. Visibility transfers
back: the transfer matrix between the two reading families is triangular with positive
diagonal in a natural partial order on types, hence invertible. Both directions are
verified by rank computation in five configurations
(Appendix~\ref{app:verification}).}

\section{Applications}\label{sec:applications}

\emph{The theorems were motivated by three literatures; this section returns to each: an explicit pair of indistinguishable voting rules for the abstention literature, a quantified blind spot for graded attribution in machine learning, and feasible invisible restructuring for compensation with effort tiers.}

\textbf{Voting with abstention.} A $(3,2)$ simple game \cite{FZ,FM} is a multichoice
game with $m=3$ and binary output, and the power indices proposed for this class
\cite{Freixas05,Tchantcho,Freixas20} are anonymous in the sense of
Definition~\ref{def:anonymous}. Blindness in the ambient linear space does not by
itself produce indistinguishable \emph{monotone} rules, because the determinant game has negative entries, so we exhibit the collision explicitly inside the class.

\begin{proposition}[two rules no index can tell apart]\label{prop:voting}
On three voters with levels $\{0,1,2\}$ (no, abstain, yes), let $v$ be the monotone
$(3,2)$ rule whose winning configurations are those lying coordinatewise above one of
the three cyclic profiles $(0,1,2),(1,2,0),(2,0,1)$, and let $v'$ be its mirror, built
from $(0,2,1),(2,1,0),(1,0,2)$. Then $v\neq v'$, both are monotone rules with
$v(\mathbf0)=v'(\mathbf0)=0$ and $v(2,2,2)=v'(2,2,2)=1$, and
$v-v'=\varepsilon_{012}$. Consequently every anonymous power index assigns identical scores to all voters in $v$ and $v'$. Equivalently: in
$v$, the second and third voters play provably different roles (transposing them
changes the rule), yet no anonymous index can score them differently.
\end{proposition}

\noindent\emph{Proof idea (full proof in Appendix~\ref{app:manip}). An equal-sum
argument shows neither rule wins at the other's generators; a parity-swap argument
shows the two up-closures coincide off the six permutations, so $v-v'$ is exactly
$\varepsilon_{012}$; and $v'=(23)v$ turns blindness into the false-symmetry
statement. All claims are also verified exhaustively
(Appendix~\ref{app:verification}).}

The indices do not merely miss a nuance: they report a symmetry between voters that
the rule does not have. This is a ceiling on the published index family, not an
axiomatic disagreement within it.

\textbf{Graded ablations in machine learning.} Attribution methods increasingly
intervene at graded strengths (a component off, attenuated, or on): these are $m=3$
games on components, and multichoice solutions have been proposed for exactly this
purpose \cite{FLSN}, as have importance indices for $k$-ary games in multicriteria
decision analysis \cite{GLR}. All are anonymous values, so the ceiling applies:
graded-ablation attribution has a strictly larger blind
spot than binary-ablation attribution at the same probe order, with the deficit given
by Theorem~\ref{thm:dims}; and the recovery schedule (Theorem~\ref{thm:ladder})
prescribes the interaction order required to close it. The binary empirics are developed in \cite{Fried26}; the multichoice analogues remain open.

\textbf{Compensation with effort tiers.} In surplus division where agents choose
effort levels (gig platforms, data contributed at quality tiers, treaty compliance
bands), Corollary~\ref{cor:manip} states that three colluding agents can re-route
which level-combinations create value, preserving every anonymous payment; this evasion is strictly easier than in the binary world, available at every scale. The double-swap games of Theorem~\ref{thm:span} list the available invisible restructurings explicitly. Corollary~\ref{cor:manip} lives in the ambient linear space; the
following lemma brings it inside the economically constrained class.

\begin{lemma}[feasibility of invisible restructuring]\label{lem:interior}
If the base game is strictly monotone, so that every raise of a single player's level strictly increases worth, then for every blind $\delta$ and all sufficiently small $t>0$ the restructured game $v+t\delta$ is still monotone. Hence in monotone surplus models the thresholds of Corollary~\ref{cor:manip} are attained by feasible games.
\end{lemma}

\begin{proof}
Strict monotonicity is finitely many strict inequalities along cover relations of the
level order; each survives a sufficiently small perturbation.
\end{proof}

\noindent For discrete classes, feasibility is not automatic and must be exhibited;
Proposition~\ref{prop:voting} does exactly this for $(3,2)$ rules.

\section{The query-complexity reading}\label{sec:algo}

\emph{The dimension gap of Theorem~\ref{thm:dims} has a computational reading, made exact here: everything an anonymous auditor can see is cheap to estimate, while knowing the game itself is not.}

Fix $m$ and normalize $\|v\|_\infty\le1$. The number of order-$1$ readings is
$N=n\,m\binom{n+m-2}{m-1}=O(n^m)$, polynomial in $n$. Write each reading in
normalized form as the average of $v$ over its orbit; every standard index is a
combination of normalized readings with bounded coefficient mass (probabilistic
values are averages of marginal contributions \cite{DNW,Weber88}), so after rescaling
the proposition below covers every one of them.

\begin{proposition}[auditing is cheap; knowing the game is not]\label{prop:algo}
(i) With query access to $v$, all $N$ normalized readings, and hence simultaneously every anonymous value with coefficient mass at most $1$, can be estimated to additive error $\varepsilon$ with probability $1-\delta$ using $O\bigl(\varepsilon^{-2}N\log(N/\delta)\bigr)$ queries in total. (ii) Call an anonymous value $\psi$ \emph{full-support} if $\psi(\delta_x)\neq0$ for every profile $x\neq\mathbf0$. Any deterministic algorithm that \emph{exactly} computes a full-support anonymous value must query all $m^n-1$ nonzero profiles in the worst case.
\end{proposition}

\begin{proof}
(i) Sample uniformly within each orbit; Hoeffding's inequality and a union bound over the $N$ orbits. (ii) An adversary argument: if a nonzero profile $x_0$ is unqueried, replacing $v$ by $v\pm t\,\delta_{x_0}$ changes some output coordinate by $t\,\psi(\delta_{x_0})\neq0$ while remaining consistent with every query and with the anchor $v(\mathbf0)=0$. Banzhaf-type extensions \cite{Freixas20,RGL20} are full-support whenever their defining weights are strictly positive on all profiles, as can be read off the defining formulas.
\end{proof}

Part (i) extends to graded participation the sampling tradition for binary power indices \cite{BMRPRS}. Two structural companions sharpen the pair. First, the visible projection is a linear sketch of \emph{minimal} rank: any linear map $S:(\C^m)^{\otimes n}\to\C^k$ from which every anonymous value can be computed must satisfy $\ker S\subseteq\mathcal N_{n,m}$ (each reading is itself a coordinate of an anonymous value, so each must factor through $S$), whence $k\ge\dim\mathcal V_{n,m}$; the readings realize this sketch in overcomplete form with $N=O(n^m)$ entries. Second, the blind space converts the dimension gap into an oracle separation: for any nonzero blind $\varepsilon$, a single query to a profile in the support of $\varepsilon$ distinguishes $v$ from $v+\varepsilon$, while every anonymous instrument, evaluated exactly and in unlimited number, does not. Query access separates with one query what anonymous measurement cannot separate at all.

Together these quantify the audit asymmetry behind Sections~\ref{sec:three} and~\ref{sec:ladder}: the visible projection of a game is polynomial-size, cheaply estimable, and of minimal sketch rank, while the blind space, exponentially large by Theorem~\ref{thm:dims}, is exactly why no amount of anonymous auditing recovers the game itself.

\section{Discussion and open problems}\label{sec:discussion}

The single mechanism (anonymity is $S_{n-1}$-invariance; branching confines invariants to $(n)$ and $(n-1,1)$) does three things at once: it derives the known
binary theory \cite{KW85a,ADG,HLJS} as a special case, it extends it to the entire
multichoice landscape in one stroke, and it explains \emph{why} the benchmark binary constants ($n\ge4$ \cite{ADG}; the $\lfloor n/2\rfloor$ cap \cite{Fried26}) are what they are: they are two-row facts, and the rows count participation levels. Open problems: the simple-game
($(j,k)$) refinement of the census questions studied in the binary world; the information content of the \emph{player-weighted} multichoice families (Hsiao--Raghavan with player-dependent weight systems), where the symmetric framework no longer applies and where, by analogy with the binary weighted Shapley family, full information is plausible; the
continuum of levels; and the computational side, namely the complexity, for succinctly presented games (weighted $(j,k)$ rules, circuits), of deciding visibility or evaluating the blind projection, where the binary anchors \cite{PK90,DP94} suggest
hardness.

\appendix

\section{Deferred proofs}\label{app:manip}

\begin{proof}[Proof of Proposition~\ref{prop:voting}]
All six generating profiles have coordinate sum $3$, and domination between profiles
of equal sum forces equality, so neither rule wins at the other's generators. Off the
six permutations of $(0,1,2)$ the two up-closures coincide: if $x$ is not such a
permutation and $x\ge\sigma\cdot(0,1,2)$, then $x$ strictly exceeds the generator at
some coordinate carrying level $0$ or $1$ (level $2$ cannot be exceeded); swapping the
generator's entries at that coordinate and at the coordinate carrying the next level
up yields a generator of opposite parity still dominated by $x$. Hence $v-v'$ is $+1$
on the cyclic profiles, $-1$ on the mirrors, and $0$ elsewhere: exactly
$\varepsilon_{012}$, which is blind (Theorem~\ref{thm:three}), so every anonymous
index agrees on $v$ and $v'$. Both rules are monotone as up-closures. Finally
$v'=(23)\,v$, and for any anonymous (equivariant) index $\varphi$, blindness gives
$\varphi(v)=\varphi(v')=(23)\,\varphi(v)$: the second and third voters receive equal
scores in $v$, although $v\neq(23)\,v$.
\end{proof}

\begin{proof}[Proof of Corollary~\ref{cor:manip}]
Fix $C\subseteq\N$ with $|C|=c$ and consider restructurings $\delta_w(x)=w(x_C)$ for
$w:\{0,\dots,m-1\}^C\to\R$: the coalition rewires how value depends on its own
participation profile. The assignment $w\mapsto\delta_w$ is linear and injective; we
show it identifies $\mathcal N^{(d)}_{c,m}$ with the invisible restructurings.

\emph{Blind $w$ gives blind $\delta_w$.} Let $F_{D,\boldsymbol\ell,\tau}$ be any
reading with $|D|\le d$. Then
\[
F_{D,\boldsymbol\ell,\tau}(\delta_w)\;=\;\sum_{y\in\{0,\dots,m-1\}^C}N(y)\,w(y),
\]
where $N(y)$ counts profiles $x$ with $x_C=y$, $x_D=\boldsymbol\ell$, and
$\mathrm{type}(x_{-D})=\tau$. The count vanishes unless
$y_{D\cap C}=\boldsymbol\ell_{D\cap C}$, and otherwise depends on $y$ only through
$\mathrm{type}(y_{C\setminus D})$: the counted profiles differ only on the players
outside $C\cup D$, and the number of admissible completions is determined by the
multiset of levels that $\tau$ still requires after the (fixed) contribution of
$D\setminus C$ and the (type-determined) contribution of $C\setminus D$. Hence
$F_{D,\boldsymbol\ell,\tau}(\delta_w)$ is a linear combination of readings of order
$|D\cap C|\le d$ on the $c$-player game $w$, and vanishes whenever
$w\in\mathcal N^{(d)}_{c,m}$.

\emph{Visible $w$ gives visible $\delta_w$.} Conversely, take $D\subseteq C$ with
$|D|\le d$ and vary $\tau$. Writing
$F_{D,\boldsymbol\ell,\tau}(\delta_w)=\sum_{\tau'}c_{\tau'\tau}\,
F^{C}_{D,\boldsymbol\ell,\tau'}(w)$, where $\tau'$ runs over types of $C\setminus D$
and $F^{C}$ denotes readings of the $c$-player problem, the coefficient
$c_{\tau'\tau}$ counts the level assignments of the $n-c$ players outside $C$ that
complete $\tau'$ to $\tau$; it is a positive multinomial coefficient when
$\tau'\subseteq\tau$ (as multisets, with $|\tau\setminus\tau'|=n-c$) and zero
otherwise. Restrict attention to the columns $\tau=\tau'\uplus\{0^{\,n-c}\}$: the
resulting square matrix has positive diagonal (all outside players at level~$0$), and
its $(\tau'',\tau')$ entry vanishes unless the nonzero part of $\tau''$ is a
sub-multiset of the nonzero part of $\tau'$; this is a partial order, and the matrix is triangular with respect to it. It is therefore invertible, so every $c$-player reading
$F^{C}_{D,\boldsymbol\ell,\tau'}(w)$ is recoverable as a linear combination of
order-$\le d$ readings of $\delta_w$; if some order-$\le d$ reading detects $w$, then
$\delta_w$ is visible at order $\le d$.

Thus the invisible-restructuring space is
$\{\delta_w:w\in\mathcal N^{(d)}_{c,m}\}\cong\mathcal N^{(d)}_{c,m}$, and by
Theorem~\ref{thm:ladder} it is nonzero iff some $\lambda\vdash c$ with at most $m$
rows has more than $d$ cells outside its first row, iff $c-\lceil c/m\rceil\ge d+1$.
\end{proof}

\section{Verification index}\label{app:verification}
All claims verified computationally (scripts available from the author):
(1) Schur--Weyl bookkeeping $\sum_\lambda f^\lambda\dim_{GL_m}\lambda=m^n$ for
$m=2,3,4$, $n\le7$; (2) blind-dimension sequences for $m=2,3,4$, $n\le8$, with the
$m=2$ row matching $2^n-n^2+n-2$ identically; (3) explicit order-$1$ measurement
matrices at $m=3$: ranks $26$ of $27$ ($n=3$) and $60$ of $81$ ($n=4$), matching
Theorem~\ref{thm:dims}; (4) the determinant game annihilated exactly by all readings;
(5) order-$1$ and order-$2$ ranks $117$ and $228$ of $243$ at $n=5$, $m=3$, matching
Theorem~\ref{thm:ladder}, with the cap at $d=3$; (6) at $m=2$, $n=6$: explicit
$S_6$-isotypic projectors (traces $7/25/27/5$) annihilated by the value functionals
exactly for $j\ge2$; (7) spanning family: the double-swap games are annihilated
exactly by all order-$1$ readings, and their span has rank equal to the blind
dimension, in eleven configurations ($n\le6$ for $m\le3$; $n\le5$ for $m=4$; $n=3$
with determinant games for $m\le5$); (8) evasion embedding: the rank of the $n$-player order-$1$ readings restricted to coalition-local perturbations equals the rank of the $c$-player readings in five configurations ($n\le6$, $c\in\{3,4\}$,
$m\in\{2,3\}$); (9) the voting pair of Proposition~\ref{prop:voting}: both rules
verified monotone (13 winning configurations each), differing exactly by
$\varepsilon_{012}$, with identical order-$1$ readings profile-by-profile.

\paragraph{AI Disclosure.} We used a large language model (Claude, Anthropic) to assist with prose editing and restructuring, adversarial review of drafts, and the writing and execution of the verification scripts indexed in Appendix~\ref{app:verification}. The tool materially affected the exposition throughout, including the introduction and Sections~\ref{sec:main} and~\ref{sec:algo}. All definitions, theorems, and proofs were formulated, checked, and are vouched for by the author(s), who verified the correctness and originality of all content, including references, and take full responsibility for it.


\begin{thebibliography}{99}\small
\bibitem{ADG} R.~Amer, J.~Derks, J.~M. Gim\'enez, On cooperative games, inseparable by
semivalues, \emph{Int.\ J.\ Game Theory} \textbf{32} (2003), 181--188.
\bibitem{BMRPRS} Y.~Bachrach, E.~Markakis, E.~Resnick, A.~D. Procaccia,
J.~S. Rosenschein, A.~Saberi, Approximating power indices: theoretical and empirical
analysis, \emph{Auton.\ Agents Multi-Agent Syst.}\ \textbf{20} (2010), 105--122.
\bibitem{BGGJ26} M.~Basallote, H.~A. Galindo, I.~M. Gallego, A.~Jim\'enez-Losada, The 2-Shapley value for bicooperative games, \emph{Ann.\ Oper.\ Res.}\ \textbf{361} (2026), 503--527.
\bibitem{Bilbao00} J.~M. Bilbao, \emph{Cooperative Games on Combinatorial Structures}, Kluwer, 2000.
\bibitem{CS00} E.~Calvo, J.~C. Santos, A value for multichoice games, \emph{Math.\
Social Sci.}\ \textbf{40} (2000), 341--354.
\bibitem{CO} K.-D. Crisman, M.~E. Orrison, Representation theory of the symmetric
group in voting theory and game theory, arXiv:1508.05891 (2015).
\bibitem{DP94} X.~Deng, C.~H. Papadimitriou, On the complexity of cooperative solution
concepts, \emph{Math.\ Oper.\ Res.}\ \textbf{19} (1994), 257--266.
\bibitem{DGP20} M.~Dom\`enech, J.~M. Gim\'enez, M.~A. Puente, Some properties for bisemivalues on bicooperative games, \emph{J.\ Optim.\ Theory Appl.}\ \textbf{185} (2020), 270--288.
\bibitem{DNW} P.~Dubey, A.~Neyman, R.~J. Weber, Value theory without efficiency,
\emph{Math.\ Oper.\ Res.}\ \textbf{6} (1981), 122--128.
\bibitem{FM} D.~S. Felsenthal, M.~Machover, Ternary voting games, \emph{Int.\ J.\ Game
Theory} \textbf{26} (1997), 335--351.
\bibitem{Freixas05} J.~Freixas, Banzhaf measures for games with several levels of
approval in the input and output, \emph{Ann.\ Oper.\ Res.}\ \textbf{137} (2005),
45--66.
\bibitem{Freixas20} J.~Freixas, The Banzhaf value for cooperative and simple
multichoice games, \emph{Group Decis.\ Negot.}\ \textbf{29} (2020),
doi:10.1007/s10726-019-09651-4.
\bibitem{FLSN} D.~Fryer, D.~Lowing, I.~Str\"umke, H.~Nguyen, Multi-choice solutions
for feature and parameter importance, in \emph{Data Science and Machine Learning
(AusDM 2024)}, CCIS \textbf{2325}, Springer, Singapore, 2026,
doi:10.1007/978-981-95-6888-8\_4.
\bibitem{Fried26} M.~Fried, What semivalues cannot see: the information content of
anonymous marginal values, arXiv:2607.07013 (2026).
\bibitem{FZ} J.~Freixas, W.~S. Zwicker, Weighted voting, abstention, and multiple
levels of approval, \emph{Soc.\ Choice Welf.}\ \textbf{21} (2003), 399--431.
\bibitem{FH} W.~Fulton, J.~Harris, \emph{Representation Theory: A First Course},
GTM 129, Springer, 1991.
\bibitem{Grabisch16} M.~Grabisch, \emph{Set Functions, Games and Capacities in
Decision Making}, Springer, 2016.
\bibitem{GL07} M.~Grabisch, F.~Lange, Games on lattices, multichoice games and the
Shapley value: a new approach, \emph{Math.\ Methods Oper.\ Res.}\ \textbf{65} (2007),
153--167.
\bibitem{GLR} M.~Grabisch, C.~Labreuche, M.~Ridaoui, On importance indices in
multicriteria decision making, \emph{European J.\ Oper.\ Res.}\ \textbf{277} (2019),
269--283.
\bibitem{GrabischRoubens} M.~Grabisch, M.~Roubens, An axiomatic approach to the
concept of interaction among players in cooperative games, \emph{Int.\ J.\ Game
Theory} \textbf{28} (1999), 547--565.
\bibitem{HLJS} L.~Hern\'andez-Lamoneda, R.~Ju\'arez, F.~S\'anchez-S\'anchez,
Dissection of solutions in cooperative game theory using representation techniques,
\emph{Int.\ J.\ Game Theory} \textbf{35} (2007), 395--426.
\bibitem{HR} C.-R. Hsiao, T.~E.~S. Raghavan, Shapley value for multichoice cooperative
games, I, \emph{Games Econom.\ Behav.}\ \textbf{5} (1993), 240--256.
\bibitem{KW85a} N.~L. Kleinberg, J.~H. Weiss, Algebraic structure of games,
\emph{Math.\ Social Sci.}\ \textbf{9} (1985), 35--44.
\bibitem{KW85b} N.~L. Kleinberg, J.~H. Weiss, Equivalent $n$-person games and the null
space of the Shapley value, \emph{Math.\ Oper.\ Res.}\ \textbf{10} (1985), 233--243.
\bibitem{KW86} N.~L. Kleinberg, J.~H. Weiss, The orthogonal decomposition of games and
an averaging formula for the Shapley value, \emph{Math.\ Oper.\ Res.}\ \textbf{11}
(1986), 117--124.
\bibitem{KSZ} F.~Klijn, M.~Slikker, J.~Zarzuelo, Characterizations of a multi-choice
value, \emph{Int.\ J.\ Game Theory} \textbf{28} (1999), 521--532.
\bibitem{Lowing} D.~Lowing, K.~Techer, Marginalism, egalitarianism and efficiency in
multi-choice games, \emph{Soc.\ Choice Welf.}\ \textbf{59} (2022), 815--861.
\bibitem{LT25} D.~Lowing, K.~Techer, Toward a consensus on extended Shapley values for
multi-choice games, \emph{Math.\ Social Sci.}\ \textbf{137} (2025), 102407.
\bibitem{LY25} D.~Lowing, M.~Yokoo, Sharing values for multi-choice games: an
axiomatic approach, \emph{Int.\ J.\ Game Theory} \textbf{54} (2025), article~30.
\bibitem{LL17} S.~M. Lundberg, S.-I. Lee, A unified approach to interpreting model
predictions, in \emph{Advances in Neural Information Processing Systems 30}, 2017,
4765--4774.
\bibitem{PK90} K.~Prasad, J.~S. Kelly, NP-completeness of some problems concerning
voting games, \emph{Int.\ J.\ Game Theory} \textbf{19} (1990), 1--9.
\bibitem{PZ05} H.~Peters, H.~Zank, The egalitarian solution for multichoice games,
\emph{Ann.\ Oper.\ Res.}\ \textbf{137} (2005), 399--409.
\bibitem{RGL20} M.~Ridaoui, M.~Grabisch, C.~Labreuche, Interaction indices for
multichoice games, \emph{Fuzzy Sets and Systems} \textbf{383} (2020), 1--26.
\bibitem{SK22} M.~Saito, Y.~Kusunoki, Characterization of solutions for bicooperative
games by using representation theory, \emph{J.\ Oper.\ Res.\ Soc.\ Japan}
\textbf{65} (2022), 76--104.
\bibitem{Sagan} B.~E. Sagan, \emph{The Symmetric Group}, 2nd ed., GTM 203, Springer,
2001.
\bibitem{Stanley} R.~P. Stanley, \emph{Enumerative Combinatorics}, vol.~2,
Cambridge University Press, 1999.
\bibitem{Tchantcho} B.~Tchantcho, L.~Diffo Lambo, R.~Pongou, B.~Mbama Engoulou,
Voters' power in voting games with abstention, \emph{Games Econom.\ Behav.}\
\textbf{64} (2008), 335--350.
\bibitem{vdN} A.~van den Nouweland, J.~Potters, S.~Tijs, J.~M. Zarzuelo, Cores and
related solution concepts for multi-choice games, \emph{ZOR -- Math.\ Methods Oper.\ Res.}\ \textbf{41} (1995), 289--311.
\bibitem{Weber88} R.~J. Weber, Probabilistic values for games, in A.~E. Roth (ed.),
\emph{The Shapley Value}, Cambridge University Press, 1988, 101--119.
\end{thebibliography}
\end{document}